\documentclass{article}
\usepackage{spconf}

\usepackage{etex}

\usepackage{relsize}

\usepackage{array}
\usepackage{booktabs,tabularx}
\usepackage{multirow}
\usepackage[binary-units=true,per-mode=symbol,detect-mode=true]{siunitx}
\usepackage{graphicx}

\usepackage{textcomp}
\usepackage[utf8]{inputenc}
\usepackage{icomma}
\usepackage{xspace}

\usepackage[tbtags]{amsmath}
\usepackage{amssymb,amsfonts,bm}
\usepackage{amsthm}
\usepackage{mathtools} 
\usepackage{dsfont}
\usepackage{mathrsfs}
\usepackage{accents}
\usepackage{empheq}
\usepackage[inline,shortlabels]{enumitem}
\usepackage{nccmath}
\usepackage{balance}

\usepackage{color}
\usepackage{calc}
\usepackage{tikz}
\usepackage{pgfplots,pgfplotstable}

\usepackage[capitalize]{cleveref}
\usepackage{refcount}

\usepackage{algorithm}
\usepackage{algpseudocode}
\makeatletter
\newcommand{\algmargin}{\the\ALG@thistlm}
\makeatother
\newlength{\whilewidth}
\algdef{SE}[parWHILE]{parWhile}{EndparWhile}[1]
  {\parbox[t]{\dimexpr\linewidth-\algmargin}{%
     \hangindent\whilewidth\strut\algorithmicwhile\ #1\ \algorithmicdo\strut}}{\algorithmicend\ \algorithmicwhile}%
\algnewcommand{\parState}[1]{\State%
  \parbox[t]{\dimexpr\linewidth-\algmargin}{\strut #1\strut}}

\usepackage{glossaries}
\usepackage{ifthen}
\usepackage{cite}
\usepackage{multibib}
\usepackage[shortcuts]{extdash}

\usepackage{comment}
\usepackage{todonotes}
\let\legacytodo\todo
\newcommand{\ruggedtodo}[2][]{\tikzexternaldisable\legacytodo[#1]{#2}\tikzexternalenable}
\renewcommand{\todo}[1]{\ruggedtodo[inline]{#1}}

\makeatletter
\def\todoref{\@ifnextchar[{\todoref@with}{\todoref@without}}
\def\todoref@without{\textbf{\color{red} [reference needed]}\xspace}
\def\todoref@with[#1]{\textbf{\color{red} [reference needed: #1]}\xspace}
\makeatother

\makeglossaries

\newacronym{cbs}{CBS}{current best solution}
\newacronym{cbv}{CBV}{current best value}
\newacronym{los}{LOS}{line-of-sight}
\newacronym{leo}{LEO}{low earth orbit}
\newacronym{iot}{IoT}{internet of things}
\newacronym{irs}{IRS}{intelligent reflecting surface}
\newacronym{socp}{SOCP}{second-order cone program}
\newacronym{soc}{SOC}{second-order cone}
\newacronym{dsl}{DSL}{digital subscriber line}
\newacronym{wsee}{WSEE}{weighted sum energy efficiency}
\newacronym{wsr}{WSR}{weighted sum rate}
\newacronym{mmwave}{mmWave}{millimeter wave}
\newacronym{dfg}{DFG}{Deutsche Forschungsgemeinschaft}
\newacronym{haec}{HAEC}{Highly Adaptive Energy-Efficient Computing}
\newacronym{hpc}{HPC}{High Performance Computing}
\newacronym{mac}{MAC}{multiple-access channel}
\newacronym{bc}{BC}{broadcast channel}
\newacronym{siso}{SISO}{single-input single-output}
\newacronym{simo}{SIMO}{single-input multiple-output}
\newacronym{miso}{MISO}{multiple-input single-output}
\newacronym{mimo}{MIMO}{multiple-input multiple-output}
\newacronym{af}{AF}{amplify-and-forward}
\newacronym{df}{DF}{decode-and-forward}
\newacronym{cf}{CF}{compress-and-forward}
\newacronym{mwrc}{MWRC}{multi-way relay channel}
\newacronym{dmmwrc}{DM-MWRC}{discrete memoryless multi-way relay channel}
\newacronym{pde}{PDE}{partial data exchange}
\newacronym{fde}{FDE}{full data exchange}
\newacronym{iid}{i.i.d.\@}{independent and identically distributed}
\newacronym{di}{DI} {difference of increasing}
\newacronym{dc}{DC}{difference of convex}
\newacronym{mm}{MM}{mixed monotonic}
\newacronym{mmp}{MMP}{mixed monotonic programming}
\newacronym{awgn}{AWGN}{additive white Gaussian noise}
\newacronym{wgn}{WGN}{white Gaussian noise}
\newacronym{awg}{AWG}{additive white Gaussian}
\newacronym{sic}{SIC}{successive interference cancellation}
\newacronym{snr}{SNR}{signal-to-noise ratio}
\newacronym{sinr}{SINR}{signal to interference plus noise ratio}
\newacronym{inr}{INR}{interference to noise ratio}
\newacronym{zf}{ZF}{zero-forcing}
\newacronym{mrt}{MRT}{maximum ratio transmission}
\newacronym{mmse}{MMSE}{minimum mean square error}
\newacronym{sud}{SUD}{single user decoding}
\newacronym{dof}{DoF}{degrees of freedom}
\newacronym{gdof}{GDoF}{generalized degrees of freedom}
\newacronym{nnc}{NNC}{noisy network coding}
\newacronym{dmn}{DMN}{discrete memoryless network}
\newacronym{csi}{CSI}{channel state information}
\newacronym{pmf}{pmf}{probability mass function}
\newacronym{dmic}{DM-IC}{discrete memoryless interference channel}
\newacronym{ic}{IC}{interference channel}
\newacronym{gic}{GIC}{Gaussian interference channel}
\newacronym{if}{IF}{interference}
\newacronym{ee}{EE}{energy efficiency}
\newacronym{gee}{GEE}{global energy efficiency}
\newacronym{tin}{TIN}{treating interference as noise}
\newacronym{snd}{SND}{simultaneous non-unique decoding}
\newacronym{sd}{SD}{simultaneous decoding}
\newacronym{hk}{HK}{Han-Kobayashi}
\newacronym{rs}{RS}{rate splitting}
\newacronym{rf}{RF}{radio frequency}
\newacronym{pa}{PA}{power amplifier}
\newacronym{lna}{LNA}{low noise amplifier}
\newacronym{lo}{LO}{local oscillator}
\newacronym{adc}{ADC}{analog-to-digital converter}
\newacronym{dac}{DAC}{digital-to-analog converter}
\newacronym{dsp}{DSP}{digital signal processing}
\newacronym{brd}{BRD}{best response dynamics}
\newacronym{br}{BR}{best response}
\newacronym{ne}{NE}{Nash equilibrium}
\newacronym{lhs}{LHS}{left-hand side}
\newacronym{rhs}{RHS}{right-hand side}
\newacronym{ran}{RAN}{radio access network}
\newacronym{qos}{QoS}{Quality of Service}
\newacronym{ngmn}{NGMN}{Next Generation Mobile Networks}
\newacronym{cap}{CAP}{Capacity Adaptation}
\newacronym{bwa}{BW}{Bandwidth Adaptation}
\newacronym{prb}{PRB}{physical resource block}
\newacronym{se}{SE}{spectral efficiency}
\newacronym{tp}{TP}{throughput}
\newacronym{bs}{BS}{base station}
\newacronym{ue}{UE}{user equipment}
\newacronym{mop}{MOP}{multi-objective optimization problem}
\newacronym{gda}{GDA}{generalized Dinkelbach's algorithm}
\newacronym{midcp}{MIDCP}{mixed integer disciplined convex programming}
\newacronym{lp}{LP}{linear program}
\newacronym{brb}{BRB}{branch reduce and bound}
\newacronym{bb}{BB}{branch and bound}
\newacronym{sit}{SIT}{successive incumbent transcending}
\newacronym{oma}{OMA}{orthogonal multiple access}
\newacronym{noma}{NOMA}{non-orthogonal multiple access}
\newacronym{rsma}{RSMA}{rate splitting multiple access}
\newacronym{sdma}{SDMA}{space division multiple access}
\newacronym{wlog}{w.l.o.g.\@}{without loss of generality}
\newacronym{lsc}{l.s.c.\@}{lower semi-continuous}
\newacronym{usc}{u.s.c.\@}{upper semi-continuous}
\newacronym{kkt}{KKT}{Karush-Kuhn-Tucker}
\newacronym{ptp}{PTP}{point-to-point}
\newacronym{ian}{IAN}{treating interference as noise}
\glsenableentrycount
\makeglossaries

\usetikzlibrary{positioning}
\usetikzlibrary{calc}
\usetikzlibrary{fit}
\usetikzlibrary{intersections}
\usetikzlibrary{decorations.pathreplacing}
\usetikzlibrary{pgfplots.fillbetween}
\usetikzlibrary{bending}
\usetikzlibrary{arrows.meta}

\usetikzlibrary{external}
\pgfkeys{/pgfplots/.cd,
	hk/.style={blue,mark=*},
	snd/.style={red,mark=square*},
	ian/.style={brown!60!black,mark=triangle*},
}

\makeatletter
\newcommand\transformxdimension[1]{
    \pgfmathparse{((#1/\pgfplots@x@veclength)+\pgfplots@data@scale@trafo@SHIFT@x)/10^\pgfplots@data@scale@trafo@EXPONENT@x}
}
\newcommand\transformydimension[1]{
    \pgfmathparse{((#1/\pgfplots@y@veclength)+\pgfplots@data@scale@trafo@SHIFT@y)/10^\pgfplots@data@scale@trafo@EXPONENT@y}
}
\makeatother

\crefname{equation}{}{}
\crefrangeformat{equation}{(#3#1#4)--(#5#2#6)}
\crefformat{footnote}{#2\footnotemark[#1]#3}

\DeclareMathOperator*\argmax{arg\,max}
\DeclareMathOperator*\argmin{arg\,min}

\undef\mod
\DeclareMathOperator\mod{mod}

\DeclareMathOperator\proj{proj}

\newcommand{\st}{\mathrm{s.\,t.}}

\let\vec\bm

\newcommand{\ubar}[1]{\underaccent{\bar}{#1}}

\allowdisplaybreaks[1]

\DeclareFontFamily{U}{mathx}{\hyphenchar\font45}
\DeclareFontShape{U}{mathx}{m}{n}{
      <5> <6> <7> <8> <9> <10>
      <10.95> <12> <14.4> <17.28> <20.74> <24.88>
      mathx10
      }{}
\DeclareSymbolFont{mathx}{U}{mathx}{m}{n}
\DeclareMathSymbol{\bigtimes}{1}{mathx}{"91}

\newtheorem{theorem}{Theorem}
\newtheorem{lemma}{Lemma}

\newtheorem{proposition}{Proposition}

\newcommand{\xqed}[1]{%
	\leavevmode\unskip\penalty9999 \hbox{}\nobreak\hfill
	\quad\hbox{\ensuremath{#1}}}

\newtheorem{XXXassumption}{Case}

\crefname{XXXassumption}{Case}{Cases}

\newtheorem{XXXexample}{Example}

\newcounter{optimizationproblem}

\newenvironment{optprob*}{\begin{equation*}\left\{\begin{aligned}}{\end{aligned}\right.\end{equation*}\ignorespacesafterend}

\makeatletter
\newcommand{\subalign}[1]{%
  \vcenter{%
    \Let@ \restore@math@cr \default@tag
    \baselineskip\fontdimen10 \scriptfont\tw@
    \advance\baselineskip\fontdimen12 \scriptfont\tw@
    \lineskip\thr@@\fontdimen8 \scriptfont\thr@@
    \lineskiplimit\lineskip
    \ialign{\hfil$\m@th\scriptstyle##$&$\m@th\scriptstyle{}##$\crcr
      #1\crcr
    }%
  }
}
\makeatother

\pgfplotscreateplotcyclelist{default}{%
	blue,mark=*\\%
	red,mark=star\\%
	teal,mark=square*\\%
	brown!60!black,mark=otimes*\\%
}

\title{Globally Optimal Beamforming for Rate Splitting Multiple Access}
\name{Bho Matthiesen$^{\star}$, Yijie Mao$^{\dagger}$, Petar Popovski$^{\ddagger \star}$, and Bruno Clerckx$^{\dagger}$\vspace{-\baselineskip}%
\thanks{This work is supported in part by the German Research Foundation (DFG) under grant EXC 2077 (University Allowance),  by the U.K. Engineering and Physical Sciences Research Council (EPSRC) under grants EP/N015312/1 and EP/R511547/1, and by the North-German Supercomputing Alliance (HLRN). }}
\address{\normalsize $^{\star}$ University of Bremen, Deptartment of Communications Engineering, Otto-Hahn-Allee 1, 28359 Bremen, Germany\\\normalsize $^{\dagger}$ Imperial College London, Deptartment of Electical and Electronic Engineering, London, United Kingdom\\\normalsize $^{\ddagger}$  Aalborg University, Department of Electronic Systems, 9220 Aalborg, Denmark}

\begin{document}
\renewcommand{\crefrangeconjunction}{--}
\renewcommand{\creflastconjunction}{, }
%\bstctlcite{IEEEexample:BSTcontrol}
\ninept
\maketitle
\begin{abstract}
	We consider globally optimal precoder design for rate splitting multiple access in Gaussian multiple-input single-output downlink channels with respect to weighted sum rate and energy efficiency maximization. The proposed algorithm solves an instance of the joint multicast and unicast beamforming problem and includes multicast- and unicast-only beamforming as special cases. Numerical results show that it outperforms state-of-the-art algorithms in terms of numerical stability and converges almost twice as fast.
\end{abstract}

  \setlength\abovedisplayskip{4.5pt plus 2.0pt minus 3.0pt}
  \setlength\belowdisplayskip{4.5pt plus 2.0pt minus 3.0pt}
  \setlength\abovedisplayshortskip{0pt plus 2.0pt}
  \setlength\belowdisplayshortskip{3.0pt plus 2.0pt minus 2.0pt}

\section{Introduction}
\cGls{rsma} is a powerful non-orthogonal transmission and robust interference management strategy for beyond 5G communication networks \cite{RSintro16bruno,mao2017eurasip,mao2019beyondDPC}. The key idea is to split each message into common and private parts and transmit them by superposition coding \cite{Han1981}. The common message is decoded by multiple users, while the private message is only decoded by the corresponding user employing \cgls{sic}. This approach allows arbitrary combinations of joint decoding and treating interference as noise by flexibly adjusting the message split.
Recent results show that \cgls{rsma} outperforms existing multiple access schemes such as \cgls{sdma}, power-domain \cgls{noma}, \cgls{oma}, and multicasting in terms of \cgls{wsr} \cite{mao2017eurasip, bruno2019wcl, mao2019TCOM} and \cgls{ee} \cite{mao2018EE, mao2019TCOM}.

This paper treats the important question of downlink \cgls{miso} beamforming for \cgls{rsma} 
with respect to \cgls{wsr} and \cgls{ee} maximization.
The corresponding optimization problem is related to joint multicast and unicast precoding that is known to be NP-hard \cite{Sidiropoulos2006,Luo2008}. Existing works on \cgls{rsma}  focus on suboptimal strategies to obtain computationally tractable algorithms \cite{hamdi2016tcom,mao2017eurasip,mao2018EE,mao2019TCOM,mao2019maxmin,Li2020a,Fu2020}. While several globally optimal algorithms for unicast beamforming \cite{Bjornson2013,Tervo2015} and multicast beamforming \cite{Lu2017} exist, joint solution methods are scarce. In particular, the procedure in \cite{Liu2017} solves the power minimization problem and \cite{Chen2018} maximizes the \cgls{wsr} for joint multicast and unicast beamforming.
All these methods are based on \cgls{bb} in combination with the \cgls{soc} transformation in \cite{Bengtsson1999}. However, as this transformation moves the complexity into the feasible set, pure \cgls{bb} methods are prone to numerical problems, see \cref{sec:alg}.
Instead, in this paper we design a \cgls{sit} \cgls{bb} algorithm to solve this beamforming problem with improved numerical stability and faster convergence. To the best of the authors knowledge, this is the first globally optimal solution algorithm for an instance of the joint unicast and multicast problem with respect to \cgls{ee} maximization. It is also the first global optimization method specifically targeted at \cgls{rsma}.

%%%%%%%%%%%%%  System Model:  %%%%%%%%%%%%%%%%%%%%%%
\section{System Model \& Problem Statement}
Consider the downlink in a wireless network where an $M$ antenna \cgls{bs} serves $K$ single-antenna users. The received signal at user $k$, $k\in\mathcal K = \{ 1, \dots, K \}$, for each channel use is
	$y_k=\vec{h}_k^{H}\vec{x}+n_k$,
where the transmit signal $\vec x\in\mathds{C}^{M\times 1}$ is subject to an average power constraint $P$, $\vec h_k$ is the complex-valued channel from the \cgls{bs} to user $k$, and $n_k$ is circularly symmetric complex white Gaussian noise with unit power at user $k$.

The transmitter employs 1-layer rate splitting \cite{mao2017eurasip,hamdi2016tcom}, i.e., it splits the message $W_k$ intended for user $k$ into a common part $W_{c,k}$ and a private part $W_{p,k}$. Then, the common messages are combined into a single message $W_c$ and these $K+1$ messages are encoded with independent Gaussian codebooks into $s_c, s_1, \dots, s_K$, each having unit power. These symbols are combined with linear precoding into the transmit signal $\vec x = \vec p_c s_c + \sum_{k\in\mathcal K} \vec p_k s_k$. The \cgls{bs} is subject to an average power constraint, i.e., $\Vert\vec p_c\Vert^2 + \sum_{k\in\mathcal K} \Vert\vec p_k \Vert^2 \le P$.

Each receiver uses \cgls{sic} to first recover $s_c$ and then $s_k$, treating all other messages as noise. Asymptotic error free decoding of $W_c$ and $W_{p,k}$ is possible if the rates of these messages satisfy $R_c \le \log(1 + \gamma_{c,k})$ and $R_{p,k} \le \log(1+\gamma_{p,k})$, with \cglspl{sinr}
\begin{equation} \label{eq: SINR cp}
	\gamma_{c,k}=\frac{|{\vec{h}}_{k}^{H}\vec{p}_{c}|^2}{\sum\nolimits_{j\in\mathcal{K}}|\vec{h}_{k}^{H}\vec{p}_{j}|^2+1},
	\enskip
	\gamma_{p,k}=\frac{|{\vec{h}}_{k}^{H}\vec{p}_{k}|^2}{\sum\nolimits_{j\in\mathcal{K}\setminus k}|\vec{h}_{k}^{H}\vec{p}_{j}|^2+1}\mathclap{.}
\end{equation}
The rate $R_{c} $ is shared across the users, where
user $k$  is allocated a portion ${C}_{k}$ corresponding to the rate of $W_{c,k}$
such that $\sum_{k \in \mathcal{K}} {C}_{k} = {R}_{c} $. Then, the total rate of user $k$ is $R_k = C_k + R_{p,k}$.

Observe that this system model includes multi-user linear precoding and multicast beamforming
as special cases. %with $\vec p_c = \vec 0$ and $\forall k\in\mathcal K : \vec p_k = \vec 0$, respectively.

\subsection{Problem Statement} \label{sec:transf}
We consider the following resource allocation problem under minimum rate $R_k^{th}$ quality of service constraints
\begin{subequations} \label{eq:srmaxequiv}
	\begin{align}
		\max_{\substack{\vec p_1, \dots, \vec p_K,\\ \vec p_c, \mathbf{c}, \vec \gamma_c, \vec \gamma_p}}\quad
		&\frac{\sum_{k\in\mathcal{K}} u_k \left(C_k+\log(1+ \gamma_{p,k})\right)}{\mu \left( \Vert \vec p_c \Vert^2 + \sum_{k\in\mathcal K} \Vert \vec p_k \Vert^2 \right) + P_c}\\
	\mbox{s.t.}\quad
	& \gamma_{c,k} \text{ and } \gamma_{p,k} \text{ as in \eqref{eq: SINR cp}}\label{eq:srmaxequiv:2} \\
	& \sum\nolimits_{k'\in \mathcal{K}}C_{k'}\leq \log(1+ \gamma_{c,k}), \forall k\in\mathcal{K} \label{eq:srmaxequiv:3} \\
	& C_k \geq \max \left\{ 0,\ R_k^{th} - \log(1+ \gamma_{p,k}) \right\}, \forall k\in\mathcal{K} \label{eq:srmaxeqiv:4} \\
	& \Vert \vec p_c \Vert^2 + \sum\nolimits_{k\in\mathcal K} \Vert \vec p_k \Vert^2 \le P \label{eq:srmaxequiv:6}
	\end{align}
\end{subequations}
with nonnegative weight vector $\mathbf{u}=[u_1,\ldots,u_K] \neq \vec 0$, nonnegative power amplifier inefficiency $\mu$, and positive static circuit power consumption $P_c$. This problem has two operational meanings: With unit weights, it maximizes the \cgls{ee} and, with $\mu = 0$, $P_c = 1$, it maximizes the \cgls{wsr}.

The following problem is equivalent to \cref{eq:srmaxequiv} and will be solved by the developed algorithm:
\begin{subequations} \label{eq:srmax}
	\begin{align}
		\max_{\substack{\vec p_c, \vec p_1, \dots, \vec p_K,\\ \mathbf{c}, \vec \gamma_p, s, \vec d, \vec e}}{} &  \frac{\sum_{k\in\mathcal{K}} u_k \left(C_k+\log(1+ \gamma_{p,k})\right)}{\mu \left( \Vert \vec p_c \Vert^2 + \sum_{k\in\mathcal K} \Vert \vec p_k \Vert^2 \right) + P_c}\\
	\mbox{s.t.}\quad
	& \sqrt{\gamma_{p,k}} \smash{\left( \sum\nolimits_{j\in\mathcal{K}\setminus k}|\mathbf{h}_{k}^{H}\mathbf{p}_{j}|^2+1 \right)^{1/2}}\!\!\! \le {\mathbf{h}}_{k}^{H}\mathbf{p}_{k}  \label{eq:srmax:1} \\
	& \sqrt s \left( \sum\nolimits_{j\in\mathcal{K}}|\mathbf{h}_{1}^{H}\mathbf{p}_{j}|^2+1 \right)^{1/2} \le {\mathbf{h}}_{1}^{H}\mathbf{p}_{c}  \label{eq:srmax:2} \\
	& \sqrt s \left( \sum\nolimits_{j\in\mathcal{K}}|\mathbf{h}_{k}^{H}\mathbf{p}_{j}|^2+1 \right)^{1/2} \le d_k, \forall k > 1 \label{eq:srmax:3} \\
	& (e_k, d_k) \in \mathcal C, \forall k > 1 \label{eq:srmax:8} \\
	& \Re\{\vec h_k^H \vec p_k\} \ge 0,\quad \Im\{\vec h_k^H \vec p_k\} = 0 \label{eq:srmax:4} \\
	& \Re\{\vec h_1^H \vec p_c\} \ge 0,\quad \Im\{\vec h_1^H \vec p_c\} = 0 \label{eq:srmax:5} \\
	& \forall k > 1: d_k \ge 0,\enskip e_k = {\mathbf{h}}_{k}^{H}\mathbf{p}_{c}\label{eq:srmax:7} \\
	& \sum\nolimits_{k\in \mathcal{K}}C_{k}\leq \log(1+ s) \label{eq:srmax:9} \\
	& \eqref{eq:srmaxeqiv:4} \text{ and } \eqref{eq:srmaxequiv:6} \label{eq:srmax:10}
	\end{align}
\end{subequations}
\begin{flalign} \label{eq:circleset}
	\text{with} && (e, d) \in \mathcal C = \{ e\in\mathds C, d \in \mathds R : d  \le | e | \}. &&
\end{flalign}
A crucial observation is that this problem is a \cgls{socp} for fixed $s$, $\vec\gamma_p$, except for constraint \cref{eq:srmax:7}. Hence, the nonconvexity of \cref{eq:srmaxequiv} is only due to the \cgls{sinr} expressions and not due to the beamforming vectors. We will exploit this partial convexity in the final algorithm to limit the numerical complexity.

\begin{proposition} \label{prop:srmax}
	Problems~\cref{eq:srmaxequiv,eq:srmax} have the same optimal value and every solution of \cref{eq:srmax} also solves \cref{eq:srmaxequiv}.
\end{proposition}
\begin{proof}
	Omitted due to space constraints. Use the \cgls{soc} reformulation from \cite{Bengtsson1999} for the \cglspl{sinr}, with additional auxiliary variables for the multicast beamformer $\vec p_c$ \cite{Lu2017}.
\end{proof}

\section{Globally Optimal Beamforming} \label{sec:alg}
Problem~\eqref{eq:srmax} is an NP-hard nonconvex optimization problem due to the multicast beamforming \cite{Sidiropoulos2006} and the power allocation in the private messages \cite{Luo2008}. Previous global optimization algorithms for similar problems rely on \cgls{bb} procedures with \cgls{socp} bounding \cite{Bjornson2013,Tervo2015,Liu2017,Chen2018}. However, this either leads to an infinite algorithm or requires the additional solution of several \cglspl{socp} to obtain a feasible point in each iteration \cite{Bjornson2013} which is required to obtain a finite algorithm. Moreover, the auxiliary \cgls{socp} that is solved in every iteration of the \cgls{bb} procedure is numerically challenging and leads to problems even with commercial state-of-the-art solvers like Mosek \cite{mosek}. This can be alleviated by the modified auxiliary problem in \cite[\S 2.2.2]{Bjornson2013} but this approach greatly increases convergence times. Instead, we design an algorithm based on the \cgls{sit} scheme \cite{Tuy2005a,Tuy2009,sit,diss} and combine it with a \cgls{brb} procedure. The resulting algorithm is numerically stable, has proven finite convergence, also solves \cgls{ee} maximization, and is the first global optimization algorithm specifically designed for \cgls{rsma}. Practically, it outperforms algorithms for similar problems as will be verified in \cref{sec:numeval}.

To better illustrate the core principles of \cgls{sit}, consider the general optimization problem
\begin{equation} \label{eq:genopt}
	\max\nolimits_{(\vec x, \vec \xi)\in\mathcal D}\enskip f(\vec x, \vec\xi) \quad\st\quad g_i(\vec x, \vec\xi) \le 0,\enskip i = 1, \dots, n
\end{equation}
with continuous, real valued functions $f, g_1, \dots, g_n$ and nonempty feasible set. Further, assume that $f$ is concave,\footnote{Although this assumption does not hold for \cref{eq:srmax}, the approach is still applicable since the sole purpose of this assumption is to obtain a convex feasible set in \cref{eq:gendual}.} $g_1, \dots, g_n$ are convex in $\vec\xi$ for fixed $\vec x$, and $\mathcal D$ is a closed convex set. Depending on the structure of $g_1, \dots, g_n$ in $\vec x$, this problem might be quite hard to solve for \cgls{bb} methods \cite{Tuy2016,sit}.\footnote{This is also true for outer approximation methods \cite{Tuy2016}.}
Instead, consider the problem
\begin{equation} \label{eq:gendual}
\min\nolimits_{(\vec x, \vec\xi)\in\mathcal D}\enskip \max\nolimits_i\{ g_i(\vec x, \vec\xi) \} \quad\st\quad f(\vec x, \vec\xi)  \ge \delta
\end{equation}
that is obtained from \cref{eq:genopt} by exchanging the objective and constraints.
If the optimal value of \cref{eq:gendual} is less than or equal to zero, the optimal value of \cref{eq:genopt} is greater than or equal to $\delta$. Instead, if the optimal value of \cref{eq:gendual} is greater than zero, the optimal value of \cref{eq:genopt} is less than $\delta$ \cite[Prop.~7]{Tuy2009}.
Hence, the optimal solution of \cref{eq:genopt} can be obtained by solving a sequence of \cref{eq:gendual} with increasing $\delta$.
Since the feasible set of \cref{eq:gendual} is closed and convex, it can be solved much easier by \cgls{brb} than \cref{eq:genopt} \cite{Tuy2009}.

The \cgls{sit} and \cgls{brb} procedures can be integrated into a single \cgls{brb} algorithm that solves \cref{eq:gendual} with low precision and updates $\delta$ whenever a point $\vec x^k$ feasible in \cref{eq:genopt} is encountered that achieves an objective value $f(\vec x^k) > \delta$.
This \cgls{brb} procedure relaxes the feasible set and subsequently partitions it in such a way that upper and lower bounds on the minimum objective value of \cref{eq:gendual} can be computed efficiently for each partition element. In particular, we use rectangular subdivision and define the initial box as
$\mathcal M_0 = [\vec r^0, \vec s^0] = \{ \vec x : r^0_i \le x_i \le s^0_i \}$
satisfying $\mathcal M_0\supseteq \proj_{\vec x} \mathcal D$.
The algorithm subsequently partitions the relaxed feasible set $\mathcal M_0$ into smaller boxes and stores the current partition of $\mathcal M_0$ in $\mathscr R_k$.
In iteration $k$, the algorithm selects a box $\mathcal M^k = [\vec r^k, \vec s^k]$ and bisects it into two new subrectangles.
 For each of these new boxes, a lower bound on the objective value is computed using a bounding function $\beta(\mathcal M)$ that computes a lower bound on the objective value of \cref{eq:gendual} with additional constraint $\vec x\in\mathcal M$. If this problem is infeasible, then $\beta(\mathcal M) = \infty$.
 To ensure convergence, the bounding needs to be consistent with branching, i.e., $\beta(\mathcal M)$ has to satisfy
\begin{equation} \label{eq:consistent}
	\beta(\mathcal M) -\ \min_{\mathclap{\substack{(\vec x, \vec\xi)\in\mathcal F,\\ \vec x\in\mathcal M}}}\ \max_i \{ g_i(\vec x, \vec\xi) \}
	\rightarrow 0 \enskip\mathrm{as}\enskip \max_{\vec x, \vec y\in\mathcal M} \Vert \vec x - \vec y \Vert \rightarrow 0,
\end{equation}
and a dual feasible point $\vec x^k\in\proj_{\vec x}\mathcal F\cap\mathcal M_k$ is required,
where $\mathcal F =  \{ \vec x\in \mathcal D : f(\vec x) \ge \delta \}$ is the feasible set of \cref{eq:gendual}.

The following lemma is essential to establish the convergence of the \cgls{sit} procedure. It follows that it can be incorporated in a \cgls{bb} procedure with pruning criterion $\beta(\mathcal M) < -\varepsilon$ and termination criterion $0 > \min\nolimits_{\vec\xi}\; g(\vec x^k, \vec\xi) \enskip\st\enskip (\vec x^k, \vec\xi)\in\mathcal F$.
\begin{lemma}[\kern -.5ex {\cite[Prop.~5.9]{diss}}] \label{lem:conv}
	Let $\varepsilon> 0$ be given and define $g(\vec x, \vec\xi) = \max_i\{g_i(\vec x, \vec\xi)\}$. Either $g(\vec x^k, \vec\xi^*) < 0$ for some $k$ and $(\vec x^k, \vec\xi)\in\mathcal F$, or $\beta(\mathcal M_k) > -\varepsilon$ for some $k$. In the former case, $(\vec x^k, \vec\xi^*)$ is a nonisolated feasible solution of \cref{eq:genopt} satisfying $f(\vec x^k, \vec\xi^*)\ge\delta$. In the latter case, no $\varepsilon$-essential feasible solution $(\vec x, \vec\xi)$ of \cref{eq:genopt} exists such that $f(\vec x, \vec\xi)\ge\delta$.
\end{lemma}
Next, we design a suitable bounding procedure that satisfies \cref{eq:consistent}.

\subsection{Bounding Procedure} \label{sec:sit:bounding}
The \cgls{sit} dual should contain all of the problem's nonconvexity in the objective function. Following the discussion in \cref{sec:transf}, the nonconvexity in \cref{eq:srmax} is due to \cref{eq:srmax:1,eq:srmax:2,eq:srmax:3,eq:srmax:8}. We obtain the \cgls{sit} dual as
\begin{subequations} \label{eq:sitdual}
	\begin{align}
		\min_{\substack{\vec p_c, \vec p_1, \dots, \vec p_K,\\ \mathbf{c}, \vec \gamma_p, s, \vec d, \vec e}}\quad & \max\Big[
			\sqrt s \left( \sum\nolimits_{j\in\mathcal{K}}|\mathbf{h}_{1}^{H}\mathbf{p}_{j}|^2+1 \right)^{1/2} - {\mathbf{h}}_{1}^{H}\mathbf{p}_{c}, \notag\\
			&\quad \max_{k > 1} \bigg\{\sqrt s \left( \sum\nolimits_{j\in\mathcal{K}}|\mathbf{h}_{k}^{H}\mathbf{p}_{j}|^2+1 \right)^{1/2}\!\!\!\!\! - d_k \bigg\}, \notag\\
			&\max_{k\in\mathcal K} \bigg\{ \sqrt{\gamma_{p,k}} \left( \sum\nolimits_{j\in\mathcal{K}\setminus k}|\mathbf{h}_{k}^{H}\mathbf{p}_{j}|^2+1 \right)^{1/2}\!\!\!\!\! - {\mathbf{h}}_{k}^{H}\mathbf{p}_{k} \bigg\}, \notag\\
		&\quad\max_{k > 1} \big\{d_k - |e_k| \big\} \quad\Big] \label{eq:sitdual:0}\\
	\mbox{s.t.}\quad
	&   \frac{\sum_{k\in\mathcal{K}} u_k \left(C_k+\log(1+ \gamma_{p,k})\right)}{\mu \left( \Vert \vec p_c \Vert^2 + \sum_{k\in\mathcal K} \Vert \vec p_k \Vert^2 \right) + P_c}\ge \delta  \label{eq:sitdual:1} \\
	& \text{\cref{eq:srmax:4,eq:srmax:5,eq:srmax:7,eq:srmax:9,eq:srmax:10}}. \label{eq:sitdual:2}
	\end{align}
\end{subequations}
Observe that \cref{eq:sitdual:1} is equivalent to the \cgls{soc}
\begin{equation*}
	\sum_{k\in\mathcal{K}} u_k \left(C_k+\log(1+ \gamma_{p,k})\right) \ge \delta \bigg( \mu \bigg( \Vert \vec p_c \Vert^2 + \sum_{k\in\mathcal K} \Vert \vec p_k \Vert^2 \bigg) + P_c \bigg)
\end{equation*}
since the denominator in \cref{eq:sitdual:1} is positive.

A bounding function $\beta(\mathcal M)$ that satisfies \cref{eq:consistent} is required.
First, observe that the objective of \cref{eq:sitdual} is increasing in $(\vec\gamma_p, s)$. Hence, a lower bound on $[\ubar{\vec\gamma}_p, \bar{\vec\gamma}_p] \times [\ubar s, \bar s]$ is obtained by setting $\vec\gamma_p = \ubar{\vec\gamma}_p$ and $s = \ubar s$ in the objective. Next,
smoothen the objective of \cref{eq:sitdual} by using the epigraph form with auxiliary variable $t$, and convert the pointwise maximum expressions to smooth constraints. Then, the new constraints $t \ge d_k - |e_k|$, for $k > 1$, are equivalent to $(e_k, d_k - t)\in\mathcal C$. This set $\mathcal C$ is nonconvex. Consistent bounding of this set is obtained using argument cuts \cite{Lu2017}, i.e., introduce auxiliary variables $\alpha_k \in[0, 2\pi]$, $k > 1$, and add the constraint $\angle e_k = \alpha_k$. The variables $\vec\alpha$ are included in the nonconvex variables handled by the \cgls{brb} solver. Then, a lower bound on the objective value of \cref{eq:sitdual} over the box $[\ubar{\vec\alpha}, \bar{\vec\alpha}]$ is obtained by replacing the constraints $d_k  \le | e_k |,\enskip \angle e_k \in [\ubar{\alpha}_k, \bar\alpha_k]$, with their convex envelope.
For $\bar\alpha_k - \ubar{\alpha}_k \le \pi$, this is
\begin{subequations}
	\label{eq:crelax}
\begin{gather}
	\sin(\ubar\alpha_k) \Re\{e_k\} - \cos(\ubar\alpha_k) \Im\{e_k\} \le 0 \label{eq:crelax:1} \\
	\sin(\bar\alpha_k) \Re\{e_k\} - \cos(\bar\alpha_k) \Im\{e_k\} \ge 0 \label{eq:crelax:2} \\
	a_k \Re\{e_k\} + b_k \Im\{e_k\} \ge (d_k - t) (a_k^2 + b_k^2) \label{eq:crelax:3}
\end{gather}
\end{subequations}
and $(e_k, d_k) \in \mathds C \times \mathds R$ otherwise \cite[Prop.~1]{Lu2017},
where
	$a_k = \tfrac{1}{2} \left( \cos(\ubar\alpha_k) + \cos(\bar\alpha_k) \right)$, and
	$b_k = \tfrac{1}{2} \left( \sin(\ubar\alpha_k) + \sin(\bar\alpha_k) \right)$.

The resulting bounding problem
depends on $\vec\gamma_p$ and $s$ only through to the constraints \cref{eq:sitdual:1,eq:srmax:9,eq:srmaxeqiv:4}, and $(\vec\gamma_p, s, \vec\alpha)\in\mathcal M$. These can be transformed into affine functions of $(\vec\gamma_p, s)$ by substituting $s' = \log(1 + s)$ and $\gamma_{p,k}' = \log( 1 + \gamma_{p,k})$. Then, these constraints are equivalent to
\begin{subequations} \label{eq:sitfeas}
	\begin{gather}
	\sum_{k\in\mathcal{K}} u_k \left( C_k + \gamma_{p,k} \right) \ge \delta \Big( \mu \Big( \Vert \vec p_c \Vert^2 + \sum_{k\in\mathcal K} \Vert \vec p_k \Vert^2 \Big) + P_c \Big) \label{eq:sitfeas:1} \\
	\sum_{k\in \mathcal{K}}C_{k}\leq s,\quad
	C_k \geq \max \left\{ 0,\ R_k^{th} - \gamma_{p,k} \right\},\enskip \forall k\in\mathcal{K} \label{eq:sitfeas:2} \\
	s \in [\log(1+\ubar s), \log(1+\bar s)] \label{eq:sitfeas:3} \\
	\gamma_{p,k} \in [\log(1+\ubar{\gamma}_{p,k}), \log(1+\bar{\gamma}_{p,k})],\quad \forall k\in\mathcal K \label{eq:sitfeas:4}
	\end{gather}
\end{subequations}
and the final bounding problem is the \cgls{socp}
\begin{subequations} \label{eq:sitbndfirst}
	\begin{align}
		\min_{\smash[b]{\substack{\vec p_c, \vec p_1, \dots, \vec p_K,\\ \mathbf{c}, \vec \gamma_p, s, \vec d, \vec e, t}}}{} & t \\
	\mbox{s.t.}\quad
	& \sqrt{\ubar{\gamma}_{p,k}} \smash[t]{\bigg( \sum\limits_{j\in\mathcal{K}\setminus k}|\mathbf{h}_{k}^{H}\mathbf{p}_{j}|^2+1 \bigg)^{1/2}}\!\!\! \le t + {\mathbf{h}}_{k}^{H}\mathbf{p}_{k}  \label{eq:sitbndfirst:1} \\
	& \sqrt{\ubar{s}} \left( \sum\nolimits_{j\in\mathcal{K}}|\mathbf{h}_{1}^{H}\mathbf{p}_{j}|^2+1 \right)^{1/2} \le t+{\mathbf{h}}_{1}^{H}\mathbf{p}_{c}  \label{eq:sitbndfirst:2} \\
	& \sqrt{\ubar{s}} \bigg( \sum\limits_{j\in\mathcal{K}}|\mathbf{h}_{k}^{H}\mathbf{p}_{j}|^2+1 \bigg)^{1/2} \le t+d_k, \forall k > 1 \label{eq:sitbndfirst:3} \\
	& \forall k\in\mathcal I_{\mathcal M} : \text{\cref{eq:crelax:1,eq:crelax:2,eq:crelax:3}} \label{eq:sitbndfirst:4} \\
	& \text{\cref{eq:srmax:4,eq:srmax:5,eq:srmax:7,eq:srmaxequiv:6,eq:sitfeas:1,eq:sitfeas:2,eq:sitfeas:3,eq:sitfeas:4}}
	\end{align}
\end{subequations}
where
$\mathcal I_{\mathcal M} = \big\{ k\in\mathcal K : k > 1 \land \max\limits_{\ubar{\alpha}, \bar{\alpha} \in \mathcal M} | \bar{\alpha}_k - \ubar{\alpha}_k | \le \pi \big\}$.
The bound $\beta(\mathcal M)$ takes the optimal value of \cref{eq:sitbndfirst} if it is feasible. Otherwise, $\beta(\mathcal M) = \infty$ otherwise.

\subsection{Feasible Point} \label{sec:sit:feas}
A dual feasible point is obtained from the solution $(\vec\gamma_p^\star, s^\star, \vec e^\star, \dots)$ of \cref{eq:sitbndfirst}
 as $(\vec\gamma_p^k, s^k, \vec\alpha^k)$ with
$\gamma_{p,i}^k = 2^{\gamma_{p,i}^\star}-1$, for $i\in\mathcal K$, $s^k = 2^{s^\star}-1$
and $\vec\alpha^k\in\proj_{\vec\alpha} \mathcal M_k = [\ubar{\vec\alpha}^k, \bar{\vec\alpha}^k]$. Numerical experiments show that the obvious choice 
$\alpha_i^k = \angle e_i^\star$ leads to very slow convergence. A much faster alternative is
$\alpha^k_i = \argmin_{\alpha\in\{\ubar{\alpha}_i^k, \bar{\alpha}_i^k\}} |\alpha - \angle e_i^\star|$.
This point is primal feasible if the optimal value of
\begin{subequations} \label{eq:sitg}
	\begin{align}
		\min_{\smash[b]{\substack{\vec p_1, \dots, \vec p_K,\\ \vec p_c, \mathbf{c}, \vec d, \vec e, t}}} {} & t
		&\mbox{s.t.}\quad
		&\text{\cref{eq:sitdual:1}, \cref{eq:srmax:4,eq:srmax:5,eq:srmax:7,eq:srmax:9,eq:srmax:10}} \vert_{\vec\gamma_p = \vec\gamma_p^k, s = s^k} \\
		&&&\text{\cref{eq:sitbndfirst:1,eq:sitbndfirst:2,eq:sitbndfirst:3}}\vert_{\ubar{\gamma}_p = \vec\gamma_p^k, \ubar s = s^k} \\
		&&&\forall i > 1: (e_i, d_i - t) \in \mathcal C,\ \angle e_i = \alpha_i^k \label{eq:sitg:3}
	%&\text{\cref{eq:sitbndfirst:4,eq:sitbndfirst:5,eq:sitbndfirst:6}} these are replaced by the above constraint
	\end{align}
\end{subequations}
is less than or equal to zero. This is an \cgls{socp} since \cref{eq:sitg:3} is affine. % constraints.

Denote the optimal solution of \cref{eq:sitg} as $(t^*, \vec c^*, \vec y^*)$. 
It can be shown that the primal objective value of $(\vec c^*, \vec y^*)$ is greater than or equal to $\delta$. This value can be further increased without impairing primal feasibility by updating $\vec c^*$ with the solution of the \cgls{lp}
	$\max_{\vec c} \sum_{k\in\mathcal K} u_k C_k \enskip\st\enskip \text{\cref{eq:sitdual:1,eq:srmax:9,eq:srmaxeqiv:4}} \vert_{\vec y^*}$.

\subsection{Reduction Procedure} \label{sec:sit:red}
The convergence criterion \cref{eq:consistent} implies that the quality of the bound $\beta(\mathcal M)$ improves as the diameter of $\mathcal M$ shrinks. Since tighter bounds lead to faster convergence, it is beneficial to reduce the size of $\mathcal M$ prior to bounding if possible at low computational cost. To ensure convergence to the global solution, it is important that the reduced box $\mathcal M' \subseteq \mathcal M$ still contains all solution candidates.

Consider the box $\mathcal M = [\ubar{\vec\gamma}_p, \bar{\vec\gamma}_p]\times[\ubar s, \bar s]\times[\ubar{\vec\alpha}, \bar{\vec\alpha}]$. Due to monotonicity, a necessary condition for the feasibility of \cref{eq:sitdual} over $\mathcal M$ is that \cref{eq:sitdual:1,eq:srmax:9,eq:srmaxeqiv:4} hold for $\bar{\vec\gamma}_p, \bar s, \bar{\vec \alpha}$.
Clearly, \cref{eq:srmax:9,eq:srmaxeqiv:4} can only hold if
\begin{equation}
	\sum\nolimits_{k\in \mathcal{I}} \left( R_k^{th} - \log(1+ \bar\gamma_{p,k}) \right) - \log(1+ \bar s)\le 0 \label{eq:feas4}
\end{equation}
with $\mathcal I = \{ k \in\mathcal K : R_k^{th} - \log(1+ \bar\gamma_{p,k}) > 0 \}$.
Similarly,
a necessary condition for \cref{eq:sitdual:1} to hold is
\begin{equation}
	\max_{k\in\mathcal K} \{ u_k \} \log(1 + \bar s) + \sum\nolimits_{k\in\mathcal{K}} u_k \log(1+ \bar\gamma_{p,k})
	\ge  \delta W \label{eq:feas5}
\end{equation}
with $W = \left( \mu \left( \min \Vert \vec p_c \Vert^2 + \sum_{k\in\mathcal K} \min \Vert \vec p_k \Vert^2 \right) + P_c \right)$,
where the minimum is such that $\vec\gamma_p\in\mathcal M$. This can be relaxed as
$\min_{\vec p_c, \dots, \vec p_K} \Vert \vec p_\kappa \Vert^2 \quad\st\quad \ubar\gamma_{p,\kappa} \le |{\mathbf{h}}_{\kappa}^{H}\mathbf{p}_{\kappa}|^2.$
From the \cgls{kkt} conditions, the optimal value of this problem is obtained as $\ubar{\gamma}_{p,\kappa}\Vert \vec h_{\kappa} \Vert^{-2}$.
Similarly, a lower bound for $\min\Vert\vec p_c\Vert^2$ is obtained as $\ubar s \max_k \Vert \vec h_k \Vert^{-2}$. Hence, 
\begin{equation} \label{eq:redW}
	W = \mu \Big( \ubar s \max_k \Vert \vec h_k \Vert^{-2} + \sum\nolimits_{k\in\mathcal K} \ubar{\gamma}_{p,k}\Vert \vec h_{k} \Vert^{-2} \Big) + P_c.
\end{equation}

Conditions \cref{eq:feas5,eq:feas4} can be used to reduce $\mathcal M$ and as a preliminary feasibility check before bounding.
For the reduction, let
$\mathcal M' =  [\ubar{\vec\gamma}_p', \bar{\vec\gamma}_p']\times[\ubar s', \bar s']\times[\ubar{\vec\alpha}, \bar{\vec\alpha}]$
and consider \cref{eq:feas5}. Every dual feasible $\gamma_{p,\kappa}\in\mathcal M$ satisfies
	$W \delta
	\le U - u_\kappa \log(1+\bar\gamma_{p,\kappa})  + u_\kappa \log(1+ \gamma_{p,\kappa})$,
where $U$ is the \cgls{rhs} of \cref{eq:feas5}. Hence, every dual feasible $\gamma_{p,\kappa}$ satisfies $\gamma_{p,\kappa} \ge 2^{\frac{W \delta - U}{u_\kappa}} (1+ \bar\gamma_{p,\kappa}) - 1$. Similarly, let $V$ be the \cgls{lhs} of \cref{eq:feas4}. From this condition, we see that every dual feasible $\gamma_{p,\kappa}$ satisfies
$\gamma_{p,\kappa} \ge 2^V (1+ \bar\gamma_{p,\kappa}) - 1$,  for $\kappa\in\mathcal I$, and $\gamma_{p,\kappa} \ge 2^{V + R_\kappa^{th}} - 1$, for $\kappa\notin\mathcal I$.
Thus, the lower bound for $\gamma_{p,k}$ can be reduced to $\ubar{\gamma}_{p,k}' = \max\{ \ubar{\gamma}_{p,k}, \ubar{\gamma}_{p,k}'' \}$
without losing feasible solution candidates,
where $\ubar{\gamma}_{p,k}'' = 2^{\max\{ \frac{W \delta - U}{u_k},\ V\}} (1+ \bar\gamma_{p,k}) - 1$ if $k\in\mathcal I$, and $\max\{ \ 2^{\frac{W \delta - U}{u_k}} (1+ \bar\gamma_{p,\kappa}), 2^{V+R_k^{th}}  \} - 1$ otherwise.
Likewise, the lower bound $s$ can be reduced to
	$\ubar s' = \max\{ \ubar s,\ 2^{\max\big\{\frac{W \delta - U}{\max_{k\in\mathcal K} \{u_k\}},\ V\big\}} (1+ \bar s) - 1 \}$.

Let $W'$ be as in \cref{eq:redW}, evaluated at $(\ubar{s}', \ubar{\vec \gamma}_{p}')$, and consider \cref{eq:feas5} again. With a similar argument as before, the upper bound of $\mathcal M'$ can be reduced to
$\bar{\gamma}_{p,k}' = \min\big\{ \bar{\gamma}_{p,k},\enskip \ubar{\gamma}_{k,p}' + (\delta\mu)^{-1} \Vert \vec h_k\Vert^2 ( U - \delta W') \big\}$ and $\bar{s}' = \min\big\{ \bar s,\enskip \ubar{s}' + (\delta\mu)^{-1} \min_k \Vert \vec h_k\Vert^2 ( U - \delta W') \big\}$.
Observe that this reduction procedure may lead to $\mathcal M' = \emptyset$.% and, hence, save the complete bounding step.

\subsection{Algorithm and Convergence}
The complete algorithm is stated in \cref{alg:sitbb}. It is essentially a \cgls{brb} procedure \cite{Tuy2016,diss} that solves the \cgls{sit} dual of \cref{eq:srmax} and updates the constant $\delta$ whenever a primal feasible point is encountered.

The initial box in \ref{alg:sitbb:init} is computed as $\mathcal M_0 = [\vec 0, \bar{\vec\gamma}_p^0] \times [0, \bar s^0] \times [0, 2\pi]^{K-1}$ with $\bar\gamma_{p,k} = P \Vert \vec h_k \Vert^2$ and $\bar s = \min_{k\in\mathcal K} P \Vert \vec h_k \Vert^2$. The set $\mathscr R_k$ holds the current partition of the feasible set, $\delta_k$ is the \cgls{cbv} adjusted by the tolerance $\eta$, and $\bar{\vec x}^k$ is the \cgls{cbs}.
In \ref{alg:sitbb:branch}, the next box is selected as $\mathcal M_k$ and bisected into $\mathscr P_k$. These boxes are reduced according to \cref{sec:sit:red} in \ref{alg:sitbb:reduction}. In \ref{alg:sitbb:bounding}, bounds for each reduced box are computed, infeasibility is detected, and dual feasible points are obtained from the bounding problem. For each of these points, primal feasibility is checked in \ref{alg:sitbb:feaspoint}. If feasible, a feasible point is recovered as in \cref{sec:sit:feas} and the corresponding primal objective value is computed. If necessary, the \cgls{cbs} and $\delta_k$ are updated in \ref{alg:sitbb:incumbent}. Boxes that cannot contain primal $\varepsilon$-essential feasible solutions are pruned in \ref{alg:sitbb:prune}. The algorithm is terminated in \ref{alg:sitbb:terminate}.
%If the partition $\mathscr R_k$ contains undecided boxes, the algorithm is continued in \ref{alg:sitbb:terminate}.
\begin{theorem} \label{thm:sitbb}
	Alg.~\ref{alg:sitbb} converges in finitely many steps to a $(\varepsilon, \eta)$-optimal solution of \cref{eq:srmax} or establishes that no such solution exists.
\end{theorem}
\vspace{-5mm}
\begin{proof}%[Proof sketch]
	Omitted due to space constraints.% In principle, it follows from the discussion above, \cite[App.~C]{sit}, and \cite[Prop.~2]{Lu2017}.
\end{proof}

\begin{algorithm}[tb]
\renewcommand{\crefrangeconjunction}{--}
	\caption{\cGls{sit} Algorithm for \cref{eq:srmax}}\label{alg:sitbb}
	\footnotesize
	\centering
	\begin{minipage}{\linewidth-1em}
	\begin{enumerate}[label=\textbf{Step \arabic*},ref=Step~\arabic*,start=0,leftmargin=*]
		\item\label{alg:sitbb:init} {\bfseries (Initialization)} Set $\varepsilon, \eta > 0$. 
			Let $k=1$ and $\mathscr R_0 = \{ \mathcal M_0 \}$.
		If an initial feasible solution $\vec y^0 = (\vec p_c^0, \dots, \vec p_K^0)$ is available, set $\delta_0 = \eta + v\eqref{eq:srmaxequiv}|_{\vec y_0}$ and initialize $\bar{\vec x}^0 = (\vec\gamma_p^0, s^0, \vec\alpha^0)$ from \cref{eq: SINR cp}, $s^0 = \min_{k\in\mathcal K} \gamma_{c,k}^0$, and $\alpha^0_k = \angle \vec h_k^H \vec p_c^0$.
			Otherwise, do not set $\bar{\vec x}^0$ and choose $\delta_0 = 0$.
		\item\label{alg:sitbb:branch} {\bfseries (Branching)} Let $\mathcal M_k = [\vec r^k, \vec s^k] \in \argmin\{\beta(\mathcal M) \,|\, \mathcal M\in\mathscr R_{k-1} \}$. Bisect $\mathcal M_k$ into
			\begingroup
			\setlength\abovedisplayskip{2pt plus 2.0pt}
			\setlength\belowdisplayskip{2pt plus 2.0pt}
			\begin{align*}
				\mathcal M^- &= \{ \vec x  : r_j \le x_j \le v_j,\ r_i \le x_i \le s_i\ (i\neq j) \} \\
				\mathcal M^+ &= \{ \vec x : v_j \le x_j \le s_j,\ r_i \le x_i \le s_i\ (i\neq j) \}
			\end{align*}
			\endgroup
			where $j_k \in \argmax_j s^k_j - r^k_j$ and $\vec v^k = \frac{1}{2} (\vec s^k + \vec r^k)$. Set $\mathscr P_{k} = \{\mathcal M^k_-, \mathcal M^k_+\}$.
		\item\label{alg:sitbb:reduction} {\bfseries (Reduction)} Replace each box in $\mathscr P_k$ with $\mathcal M'$ as in \cref{sec:sit:red}.
		\item\label{alg:sitbb:bounding} {\bfseries (Bounding)} For each reduced box $\mathcal M\in\mathscr P_k$, solve \cref{eq:sitbndfirst}. If infeasible, set $\beta(\mathcal M) = \infty$. Otherwise, set $\beta(\mathcal M)$ to the optimal value of \cref{eq:sitbndfirst} and obtain a dual feasible point $\vec x(\mathcal M)$ as in \cref{sec:sit:feas}.
		\item\label{alg:sitbb:feaspoint} {\bfseries (Feasible Point)} For each $\mathcal M\in\mathscr P_k$, if $\beta(\mathcal M) \le 0$ solve \cref{eq:sitg} for $\vec x(\mathcal M)$ and denote the optimal value as $t(\vec x(\mathcal M))$. If $t(\vec x(\mathcal M)) \le 0$, $\vec x(\mathcal M)$ is primal feasible. Recover $\vec x'(\mathcal M)$ from the solution of \cref{eq:sitg} with $\vec\gamma_p'$, $s'$ as in \ref{alg:sitbb:init} and $\alpha_k' = \angle e_k^*$, $k > 1$, where $\vec e^*$ is from the optimal solution of \cref{eq:sitg}. Update $\vec c^*$ as in \cref{sec:sit:feas} and compute the primal objective value $f(\mathcal M)$. If $\beta(\mathcal M) > 0$ or $t(\vec x(\mathcal M)) > 0$, set $f(\mathcal M) = -\infty$.
		\item\label{alg:sitbb:incumbent} {\bfseries (Incumbent)} Let $\mathcal M' \in \argmin\{ f(\mathcal M) : \mathcal M\in\mathscr P_k \}$. If $f(\mathcal M') > \delta_{k-1} - \eta$, set $\bar{\vec x}^k = \vec x'(\mathcal M')$ and $\delta_k = f(\mathcal M') + \eta$. Otherwise, set $\bar{\vec x}^k = \bar{\vec x}^{k-1}$ and $\delta_k = \delta_{k-1}$.
		\item\label{alg:sitbb:prune} {\bfseries (Pruning)} Delete every $\mathcal M\in\mathscr P_k$ with $\beta(\mathcal M) > -\varepsilon$ and collect the remaining sets in $\mathscr P_k'$. Set $\mathscr R_k = \mathscr P_k'\cup(\mathscr R_{k-1}\setminus\{\mathcal M_k\})$.
		\item\label{alg:sitbb:terminate} {\bfseries (Termination)} Terminate if $\mathscr R = \emptyset$: If $\bar{\vec x}^k$ is not set, then \cref{eq:srmax} is $\varepsilon$-essential infeasible; else $\bar{\vec x}^k$ is an essential $(\varepsilon, \eta)$-optimal solution of \cref{eq:srmax}. Otherwise, update $k\gets k+1$ and return to \ref{alg:sitbb:branch}.
	\end{enumerate}
	\end{minipage}
\end{algorithm}

\section{Numerical Evaluation} \label{sec:numeval}
As most numerical problems of similar state-of-the-art algorithms arise from the multiple unicast beamforming problem, i.e., where $\vec p_c = \vec 0$, we evaluate the performance of the algorithm for this case. In particular, we have generated 100 random i.i.d.\ channel realizations and solved \cref{eq:srmaxequiv} for $u_k = 1$, $\mu = 0$, $P_c = 0$, $R_{k}^{th} = 0$, $\frac{P}{\mathrm{dB}} = -10, -5, \dots, 20$, and $K = M \in \{ 2, 3, 4 \}$. This results in 700 problem instances per $K$. As baseline comparison and verification, we chose the straightforward \cgls{bb} implementation of this problem \cite{Bjornson2013,Tervo2015} (``BB'') and its variant with modified bounding problem from \cite[\S 2.2.2]{Bjornson2013} (``BB2''). For $K = 2$, BB2 stalled in 364 problem instances, while the other algorithms solved all problems. For $K = 3$, BB2 stalled 146 times and BB failed 13$\times$ due to numerical problems of the convex solver. Finally, for $K = 4$, BB did not solve a single problem instance due to numerical issues and BB2 stalled in 27 instances. Moreover, \cref{alg:sitbb} and BB2 did not solve the problem withing 60 minutes in 4 and 60 instances, respectively. Average computation times on a single core of an Intel Cascade Lake Platinum 9242 CPU are reported in \cref{tab}. It can be observed that the proposed \cref{alg:sitbb} is more efficient than the two baseline algorithms especially when more users are in the system.
Moreover, the joint beamforming problem, i.e., with $\vec p_c \neq \vec 0$, was solved by \cref{alg:sitbb} for $K = 2$ with mean and median run times of \SI{942}{\second} and \SI{2786}{\second}. However, 23 instances were not solved within 12 hours.

Observe from the discussion in \cref{sec:alg} that the complexity scales with $O(\exp(2 K))$ in the number of users and polynomially in the number of antennas $M$. Hence, no noticeable changes in the reported run times are to be expected by varying $M$.

\begin{table}[tb]
	\centering
	\footnotesize
	\begin{tabular}{lccc}
		\toprule
		 & $K = 2$ & $K = 3$ & $K = 4$\\
		 \midrule
		 Alg.~\ref{alg:sitbb} & \SI{0.175}{\second} / \SI{0.099}{\second} & \SI{4.579}{\second} / \SI{1.959}{\second} & \SI{334.8}{\second} / \SI{126.3}{\second} \\
		 BB  & \SI{0.173}{\second} / \SI{0.091}{\second} & \SI{7.605}{\second} / \SI{2.606}{\second} & --- \\
		 BB2 & \SI{42.41}{\second} / \SI{2.380}{\second} & \SI{158.5}{\second} / \SI{12.42}{\second} & \SI{704.1}{\second} / \SI{265.8}{\second}\\
		\bottomrule
	\end{tabular}
	\caption{Mean / median run times to obtain the optimal solution. Problem instances where not all algorithms converged are ignored.}
	\label{tab}
\end{table}

\section{Conclusions}
We developed the first global optimization algorithm to solve \cgls{miso} downlink beamforming for \cgls{rsma} with respect to \cgls{wsr} and \cgls{ee} maximization. This problem is an instance of joint multicast and unicast beamforming and also solves these problems separately. The algorithm is numerically stable and outperforms state-of-the-art multiple unicast beamforming algorithms considerably.

%\bibliography{IEEEabrv,IEEEtrancfg,bibliography.bib}
\clearpage
\balance
\bibliographystyle{IEEEbib}
\bibliography{IEEEabrv,bibliography}
\end{document}